\documentclass[a4paper,USenglish,numberwithinsect]{article}

\usepackage[utf8]{inputenc}
\usepackage{fullpage}
\usepackage{microtype}
\usepackage{amsfonts}
\usepackage{amssymb}
\usepackage{amsmath}
\usepackage{amsthm}
\usepackage{paralist}
\usepackage{graphicx}
\usepackage{wrapfig}
\usepackage{color}
\usepackage{verbatim}
\usepackage{xspace}
\usepackage{chngcntr}
\usepackage{authblk}
\usepackage[unicode]{hyperref}

\hypersetup{
   pdftitle={Online Packet Scheduling with Bounded Delay and Lookahead},
   pdfauthor={M. Böhm, M. Chrobak, Ł. Jeż, F. Li, J. Sgall, and P. Veselý}
}

\newcommand{\ToggleH}{\textsf{ToggleH}}
\newcommand{\myparagraph}[1]{{\smallskip\noindent{\bf #1}}}
\newcommand{\emparagraph}[1]{{\smallskip\noindent\emph{#1}}}
\newcommand{\etal}{{\em et al.\/}}

\newcommand{\mycase}[1]{\mbox{{\underline{Case #1}}:\/}}

\newcommand{\bart}{{\bar t}}

\newcommand{\EDF}{{\textrm{EDF}}}
\newcommand{\BDPS}{{\textsf{PacketScheduling}}}
\newcommand{\OPT}{$\vari{OPT}$\xspace}

\let\epsilon=\varepsilon
\def\I{\it\aftergroup\/}

\def\text#1{\hbox{#1}}

\def\OPT{$\textsf{OPT}$\xspace}
\def\ALG{$\textsf{ALG}$\xspace}
\def\LC{{\sc CompareWithBias}\xspace}

\newcommand\LCalpha{{\sc CompareWithBias}$(\alpha)$\xspace}

\newcommand{\thinnegspace}{{\hspace{-0.02in}}}
\newcommand{\scdot}{{\thinnegspace\cdot\thinnegspace}}

\newcommand{\braced}[1]{{ \left\{ #1 \right\} }}

\newcommand{\hatR}{\widehat{R}}

\newcommand{\calA}{{\cal A}}

\newcommand{\half}{{\textstyle\frac{1}{2}}}
\newcommand{\onehalf}{{\textstyle\frac{1}{2}}}

\newcommand{\threehalves}{{\textstyle\frac{3}{2}}}

\newcommand{\onefourth}{{\textstyle\frac{1}{4}}}

\newcommand{\onesixth}{{\textstyle\frac{1}{6}}}

\newcommand{\ignore}[1]{}

\newtheorem{theorem}{Theorem}[section]
\newtheorem{lemma}[theorem]{Lemma}

\title{Online Packet Scheduling with Bounded Delay and Lookahead%
\footnote{M. B{\"o}hm,  J. Sgall, and P. Vesel\'{y} were supported by
project 14-10003S of GA \v{C}R and by the GAUK project 548214.
M.~Chrobak was supported by NSF grants CCF-1217314 and CCF-1536026.
{\L}.~Je{\.z} was supported by NCN grant DEC-2013/09/B/ST6/01538.
F. Li was supported by NSF grant CCF-1216993.}
}

\author[1]{Martin Böhm}
\affil[1]{
Computer Science Institute of Charles University,
Prague, Czech Republic.
\texttt{\{bohm,sgall,vesely\}@iuuk.mff.cuni.cz}
}

\author[2]{Marek Chrobak}
\affil[2]{
Department of Computer Science and Engineering,
University of California, Riverside, USA.
\texttt{marek@cs.ucr.edu}.
}

\author[3]{Łukasz Jeż}
\affil[3]{
Institute of Computer Science,
University of Wrocław, Poland.
\texttt{lje@cs.uni.wroc.pl}.
}

\author[4]{Fei Li}
\affil[4]{
Department of Computer Science,
George Mason University, USA.
\texttt{lifei@cs.gmu.edu}.
}

\author[1]{Jiří Sgall}

\author[1]{Pavel Veselý}

\begin{document}

\maketitle

\begin{abstract}
We study the \emph{online bounded-delay packet scheduling problem (\BDPS)}, where 
packets of unit size arrive at a router over time and need to be transmitted 
over a network link. Each packet has two attributes: 
a non-negative weight and a deadline for its transmission.
The objective is to maximize the total weight of the transmitted packets.
This problem has been well studied in the literature, yet its
optimal competitive ratio remains unknown: the best upper bound is 
$1.828$~\cite{englert_suppressed_packets_07}, still quite far from the best lower bound of  
$\phi \approx 1.618$~\cite{hajek_unit_packets_01,andelman_queueing_policies_03,chin_partial_job_values_03}.

In the variant of {\BDPS} with \emph{$s$-bounded instances}, each packet can be
scheduled in at most $s$ consecutive slots, starting at its release time.
The lower bound of $\phi$ applies even to the special case of $2$-bounded
instances, and a $\phi$-competitive algorithm for $3$-bounded
instances was given in~\cite{chin_weighted_throughput_06}.
Improving that result, and addressing a question posed by Goldwasser~\cite{goldwasser_survey_10},
we present a $\phi$-competitive algorithm for \emph{$4$-bounded} instances.

We also study a variant of {\BDPS} where an online
algorithm has the additional power of \emph{1-lookahead}, knowing at
time $t$ which packets will arrive at time $t+1$. For {\BDPS} with 1-lookahead 
restricted to $2$-bounded instances, we present an online
algorithm with competitive ratio $\half (\sqrt{13} - 1) \approx 1.303$ 
and we prove a nearly tight lower bound of $\onefourth (1 + \sqrt{17}) \approx 1.281$.
\end{abstract}


\section{Introduction}\label{sec:intro}



\myparagraph{Background.}
Optimizing the flow of packets across an IP network gives rise to a plethora of
challenging algorithmic problems. In fact, even  
scheduling packet transmissions from a router across a specific network 
link can involve non-trivial tradeoffs.
Several models for such tradeoffs have been formulated, depending on the
architecture of the router, on characteristics of the packets,
and on the objective function.

In the model that we study in this paper, each packet has two attributes: 
a non-negative weight and a deadline for its transmission. The time is assumed to
be discrete (slotted), and only one packet can be sent in each slot.
The objective is to maximize the total weight of the transmitted packets. 
We focus on the online setting, where at each time step the router needs to choose
a pending packet for transmission, without the knowledge about future packet arrivals.
This problem, which we call \emph{online bounded-delay packet scheduling problem} ({\BDPS}),
was introduced by Kesselman~\etal~\cite{kesselman_buffer_overflow_04} as a 
theoretical abstraction that captures the constraints and objectives of packet
scheduling in networks that need to provide quality of service (QoS) guarantees.
The combination of deadlines and weights is used to model packet priorities.

In the literature, the {\BDPS} problem is sometimes referred to as 
\emph{bounded-delay buffer management in QoS switches}. It can also be formulated
as the job-scheduling problem $1|p_j = 1,r_j|\sum w_jU_j$, 
where packets are represented by unit-length
jobs with deadlines, with the objective to maximize the weighted throughput.

A router transmitting packets across a link needs to make scheduling decisions 
on the fly, based only on the currently available information. This motivates the
study of online competitive algorithms for {\BDPS}.
A simple online greedy algorithm that always schedules the heaviest pending packet is
known to be $2$-competitive~\cite{hajek_unit_packets_01,kesselman_buffer_overflow_04}. 
In a sequence of papers~\cite{chrobak_improved_buffer_04,englert_buffer_management_09,li_optimal_agreeable_05,englert_suppressed_packets_07},
this ratio was gradually improved, and the best
currently known ratio is $1.828$~\cite{englert_suppressed_packets_07}.
The best lower bound, widely believed to be the optimal ratio, is
$\phi = (1 + \sqrt{5}) / 2 \approx 1.618$~\cite{hajek_unit_packets_01,andelman_queueing_policies_03,chin_partial_job_values_03}.
Closing the gap between these two bounds is one of the most intriguing open problems in
online scheduling.


\myparagraph{$s$-Bounded instances.}
In an attempt to bridge this gap, restricted models 
%
%
have been studied.
In the \emph{$s$-bounded} variant of {\BDPS}, each packet must be scheduled
within $k$ consecutive slots, starting at its
release time, for some $k\leq s$ possibly depending on the packet.
%
%
The lower bound of $\phi$ from~\cite{hajek_unit_packets_01,andelman_queueing_policies_03,chin_partial_job_values_03}
holds even in the $2$-bounded case.
A matching $\phi$-competitive algorithm was given 
Kesselman~{\etal}~\cite{kesselman_buffer_overflow_04} for $2$-bounded
instances and by Chin~{\etal}~\cite{chin_weighted_throughput_06}
for $3$-bounded instances. 
Both results are based on the algorithm $\EDF_\alpha$, with $\alpha=\phi$, which always schedules the earliest-deadline
packet whose weight is at least the weight of the heaviest pending packet divided by $\alpha$
(ties are broken in favor of heavier packets). $\EDF_\phi$ is not $\phi$-competitive
for $4$-bounded instances; however, a different choice of $\alpha$
yields a $1.732$-competitive algorithm for the $4$-bounded case~\cite{chin_weighted_throughput_06}.

\emparagraph{Our contribution.} We present a
$\phi$-competitive online algorithm for {\BDPS} restricted to $4$-bounded instances, matching the
lower bound of $\phi$ (see Section~\ref{sec:4bounded}).
This improves the results from~\cite{chin_weighted_throughput_06}
and answers the question posed by Goldwasser
in his SIGACT~News survey~\cite{goldwasser_survey_10}. 


\myparagraph{Algorithms with 1-lookahead.}
In Sections~\ref{sec:lookaheadalgo} and \ref{sec:lookaheadlb}, we
investigate a variant of {\BDPS} where an
online algorithm is able to learn at time $t$ which packets will
arrive by time $t+1$. This property is known as
\emph{1-lookahead}. From a practical point of view, 1-lookahead corresponds
to the situation in which a router can see the packets that are just
arriving to the buffer and that will be available for transmission in the next time slot.

The notion of lookahead is quite natural and it has appeared in the online algorithm
literature for paging~\cite{albers_lookahead_97}, scheduling~\cite{motwani_lookahead_assembly_lines_98}
and bin packing~\cite{grove_bin_packing_lookahead_95} since the 1990s.
Ours is the first paper, to our knowledge, that considers 
lookahead in the context of packet scheduling.

\emparagraph{Our contributions.}
We provide two results about {\BDPS} with 1-lookahead, restricted to $2$-bounded instances.
First, in Section~\ref{sec:lookaheadalgo}, we present
an online algorithm for this problem with competitive ratio of $\half (\sqrt{13} - 1) \approx 1.303$.
Then, in Section~\ref{sec:lookaheadlb}, we give a lower bound of $\onefourth(1 +
\sqrt{17}) \approx 1.281$ on the competitive ratio of algorithms with 1-lookahead
which holds already for the $2$-bounded case.


\section{Definitions and Notation}\label{sec:definitions}



\myparagraph{Problem statement.}
Formally, we define the {\BDPS} problem as follows.
The instance is a set of packets, with each packet $p$ specified by a
triple $(r_p,d_p,w_p)$, where $r_p$ and $d_p\ge r_p$ are integers representing
the \emph{release time} and \emph{deadline} of $p$, and $w_p\ge 0$ is a real number
representing the \emph{weight} of $p$. Time is discrete, divided into
unit \emph{time slots}, also called \emph{steps}.
A \emph{schedule} assigns time slots to some subset of packets such that
(i) any packet $p$ in this subset is assigned a slot in the interval $[r_p,d_p]$, 
and (ii) each slot is assigned to at most one packet.
The objective is to compute a schedule that maximizes the total weight of the
scheduled packets, also called the {\em profit}. 

In the \emph{$s$-bounded} variant of {\BDPS}, we assume that each
packet $p$ in the instance satisfies $d_p\le r_p+s-1$. In other words,
this packet must be scheduled within $k_p$ consecutive slots, starting
at its release time, for some $k_p\leq s$.

%
%


\myparagraph{Online algorithms.}
In the online variant of {\BDPS}, which is the focus of our work, at any time $t$
only the packets released at times up to $t$ are revealed. Thus an online
algorithm needs to decide which packet to schedule at time $t$ (if any)
without any knowledge of packets released after time $t$.

As is common in the area of online optimization, we measure the
performance of an online algorithm $\calA$ by its competitive
ratio. An algorithm is $R$-competitive if, for all instances, the
total weight of the optimal schedule (computed offline) is at most $R$
times the weight of the schedule computed by $\calA$.

We say that a packet is \emph{pending} for an algorithm at time $t$, 
if $r_p \le t \le d_p$ and $p$ is not scheduled before time $t$.
A (pending) packet $p$ is \emph{expiring} at time $t$ if $d_p = t$,
that is, it must be scheduled now or never.
A packet $p$ is \emph{tight} if $r_p = d_p$; thus $p$ is 
expiring already at its release time.


\myparagraph{Algorithms with 1-lookahead.}  
In Sections~\ref{sec:lookaheadalgo} and \ref{sec:lookaheadlb}, we
investigate the {\BDPS} problem \emph{with 1-lookahead}. With
1-lookahead, the problem definition changes so that at time $t$, an
online algorithm can also see the packets that will be released at
time $t+1$, in addition to the pending packets.  Naturally, only a
pending packet can be scheduled at time $t$.


\myparagraph{Other terminology and assumptions.}
We will make several assumptions about our problem that do not affect the
generality of our results. First, we can assume that all packets have different
weights. Any instance can be transformed into an instance with distinct weights
through infinitesimal perturbation of the weights, without affecting the
competitive ratio. Second, we assume that at each step there is at least one
pending packet. (If not, we can always release a tight packet of
weight $0$ at each step.)

We define the \emph{earliest-deadline relation} on packets, or \emph{canonical
ordering}, denoted $\prec$, where $x\prec y$ means that either $d_x <
d_y$ or $d_x = d_y$ and $w_x > w_y$ (so the ties are broken in favor
of heavier packets).  At any step $t$, the algorithm maintains the
earliest-deadline relation on the set of its pending packets.
Throughout the paper, ``earliest-deadline packet'' means 
the earliest packet in the canonical ordering.

Regarding the adversary (optimal) schedule, we can assume that it
satisfies the following \emph{earliest-deadline property}: if packets $p$, $p'$  are
scheduled in steps $t$ and $t'$, respectively, where 
$r_{p'} \le t  < t' \le d_p$  (that is, $p$ and $p'$
can be swapped in the schedule without violating their
release times and deadlines), then $p\prec p'$. This can be rephrased in
the following useful way: at any step, the optimum schedule transmits
the earliest-deadline packet among all the pending packets
that it transmits in the future.


\section{An Algorithm for 4-bounded Instances}\label{sec:4bounded}



In this section, we present a $\phi$-competitive algorithm for
$4$-bounded instances.  Ratio $\phi$ is of course optimal~\cite[see
  also
  Section~\ref{sec:intro}]{hajek_unit_packets_01,andelman_queueing_policies_03,chin_partial_job_values_03}.
Up until now, the best competitive ratio for $4$-bounded instances was
$\sqrt{3}\approx 1.732$, achieved by algorithm $\EDF_{\sqrt{3}}$ in
\cite{chin_weighted_throughput_06}. Our algorithm can be seen as a
modification of $\EDF_\phi$, which under certain conditions schedules
a packet lighter than $w_h/\phi$ where $h$ is the heaviest pending
packet.

We remark that our algorithm uses memory; in particular, it marks one
pending packet under certain conditions. It is an interesting question
whether there is a memoryless $\phi$-competitive algorithm for $4$-bounded
instances.


\myparagraph{Algorithm~$\ToggleH$.}
The algorithm maintains one mark that may be assigned to one of the
pending packets. For a given step $t$, we choose the following packets from
among all pending packets:
\begin{description}
	\setlength\itemsep{-1pt}
	\item{$h=$} the heaviest packet,
	\item{$s=$} the second-heaviest packet,
	\item{$f=$} the earliest-deadline packet with $w_f\ge w_h/\phi$, and
	\item{$e=$} the earliest-deadline packet with $w_e \ge w_h/\phi^2$.
\end{description}

We then proceed as follows:
\begin{tabbing}
aaa \= aaa \= aaa \= aaa \= aaa \= aaa \= \kill
\> \textbf{if} ($h$ is not marked) $\vee$  ($w_s \ge w_h/\phi$) $\vee$ ($d_e > t$)
\\
\>\> schedule $f$
\\
\> \> \textbf{if} there is a marked packet \textbf{then} unmark it
\\
\>\> \textbf{if} ($d_h =t+3$) $\wedge$ ($d_f = t+2$) \textbf{then} mark $h$
\\
\> \textbf{else} // ($h$ is marked) $\wedge$ ($w_s < w_h/\phi$) $\wedge$ ($d_e = t$)
\\
\>\> schedule $e$
\\
\> \>  unmark $h$
\end{tabbing}
Note that when $f\neq h$, then the algorithm will always schedule $f$. This is because
in this case $f$ is a candidate for $s$, so the condition $w_s\ge w_h/\phi$ holds.
The algorithm never specifically chooses $s$ for scheduling -- it is only
used to determine if there is one more relatively heavy pending packet other than $h$. (But
 $s$ \emph{may} get scheduled if it so happens that $s =f$ or
$s = e$.) 
Note also that, if $e\neq f$, then $e$ is scheduled only in a very
specific scenario, when all of the following hold:
$e$ is expiring, $h$ is marked, and $w_s < w_h/\phi$.
%


\myparagraph{Intuition.} 
Let us give a high-level view of the analysis using charging schemes
and an example that motivates both our algorithm and its analysis.
The example consists of four packets $j,k,f,h$ released in step $1$,
with deadlines $1,2,3,4$ and weights
$1-\varepsilon,1-\varepsilon,1,\phi$ for a small $\varepsilon>0$,
respectively. The optimum schedules all packets.

Algorithm $\EDF_\phi$ performs only $f$-steps; in our example it
schedules $f$ and $h$ in steps $1$ and $2$, while $j$ and $k$ are
lost. Thus the ratio is larger than $\phi$. (In fact, after optimizing
the threshold and the weight of $h$, this is the tight example for
$\EDF_{\sqrt{3}}$ on 4-bounded instances.)  $\ToggleH$ avoids this
example by performing $e$-step in step $2$ and scheduling $k$ which
has the role of $e$ and $s$ in the algorithm.

This example and its variants are also important for our analysis. We
analyze the algorithms by charging schemes, where the weight of each
packet scheduled by the adversary is charged to one or more of the
slots of the algorithm's schedule. If the weight charged to each slot
is at most $R$ times the weight of the packet scheduled by the
algorithm in that slot, the algorithm is $R$-competitive. In the case
of $\EDF$, we charge the weight of each packet $j$ scheduled by the
adversary at time $t$ either fully to the step where $\EDF$ schedules
$j$, if it is before $t$, or fully to step $t$ otherwise. In our
example, the weight charged to step $1$ is $2-\varepsilon$ while
$\EDF$ schedules only weight $1$, giving the ratio $2$. Considering
steps $1$ and $2$ together leads to a better ratio and after balancing
the threshold it gives the tight analysis of $\EDF_{\sqrt{3}}$.

Our analysis of $\ToggleH$ is driven by the variants of the example
above where step~$2$ is an $f$-step. This may happen in several
cases. One case is if in step~$2$ another packet $s$ with $w_s\geq
w_h/\phi$ arrives. If $s$ is not scheduled in step $2$, then $s$ is
pending in step~$3$, thus $\ToggleH$ schedules a relatively heavy
packet in step~$3$, and we can charge a part of the weight of $f$,
scheduled in step $3$ by the adversary, to step~$3$.
This motivates the definition of regular up and back charges
below and corresponds to Case~5.1 in the analysis.  Another case is
when the weight of $k$ is changed to $1/\phi-\varepsilon$. Then $\ToggleH$
performs an $f$-step because $k$ is not a candidate for $e$, thus the
role of $e$ is taken by the non-expiring packet $h$. However, then the
weight of the four packets charged to steps $1$ and $2$ in the way
described above is at most $\phi$ times the weight of $f$ and $h$;
this corresponds to Case~5.2 of the analysis. Lemma~\ref{lem:aux}
gives a subtle argument showing that in the 4-bounded case essentially
these two variants of our example are the only difficult
situations. Finally, in the original example, $\ToggleH$ schedules $k$
in step $2$ which is an $e$-step. Then again $h$ is a pending heavy
packet and we can charge some weight of $f$ to step $3$. Intuitively
it is important that an $e$-step is performed only in a very specific
situation where it is guaranteed that $h$ can be scheduled in the next
two steps (as it is marked) and that there is no other packet of
comparable weight due to the condition $w_s<w_h/\phi$. Still, there is
a case to be handled: If more packets arrive in step $3$, it is also
possible that the adversary schedules $h$ already in step $2$ and we
need to redistribute its weight. This case motivates the definition of
the special up and back charges below.

\begin{theorem}\label{thm:toggle-h}
	Algorithm~$\ToggleH$ is $\phi$-competitive on $4$-bounded instances.
\end{theorem}

\begin{proof}
Fix some optimal adversary schedule. Without loss of generality, we
can assume that this schedule satisfies the earliest-deadline property
(see Section~\ref{sec:definitions}).

We have two types of packets scheduled by Algorithm~$\ToggleH$:
\emph{f-packets}, scheduled using the first case, and
\emph{e-packets}, scheduled using the second case. Similarly, we refer
to the steps as \emph{$f$-steps} and \emph{$e$-steps}.

Let $t$ be the current step. By $h$, $f$, $e$, and $s$ we denote the
packets from the definition of $\ToggleH$. By $j$ we denote the packet
scheduled by the adversary. By $h'$ and $h''$ we denote the heaviest
pending packets in steps $t+1$ and $t+2$, respectively. We use the
same convention for packets $f$, $e$, $s$, and $j$.

Our analysis uses a new charging scheme which we now define. The
adversary packet $j$ scheduled in step $t$ is charged according to the
first case below that applies:

\begin{enumerate}
	
\item 
If $t$ is an $e$-step and $j=h$, we charge $w_h/\phi$ to step $t$
and $w_h/\phi^2$ to step $t-1$. We call these charges a
\emph{special up charge} and a \emph{special back charge},
respectively. Note that the total charge is equal to $w_h=w_j$.

\item 
If $j$ is pending for $\ToggleH$ in step $t$, charge $w_j$ to
step $t$.  We call this charge a \emph{full up charge}.

\item 
Otherwise $j$ is scheduled before step $t$. We charge
$w_h/\phi^2$ to step $t$ and $w_j-w_h/\phi^2$ to the step where
$\ToggleH$ scheduled $j$.  We call these charges a \emph{regular up
  charge} and a \emph{regular back charge}, respectively. We point out that
the regular back charge may be negative, but this causes no problems in
the proof. 

\end{enumerate}

We start with an easy observation that we use several times throughout the proof.


\begin{lemma}\label{l:up}
If an $f$-step $t$ receives a regular back charge, then the up charge it
receives is less than $w_h/\phi$.
\end{lemma}

\begin{proof}
For a regular up charge the lemma is trivial (with a slack of a factor
of $\phi$). For a full up charge, the existence of a back charge
implies that the adversary schedules $f$ after $j$, thus the
earliest-deadline property of the adversary schedule implies that
$j\prec f$, as both $j$ and $f$ are pending for the adversary at
$t$. Thus $\ToggleH$ would schedule $j$ if $w_j\geq
w_h/\phi$. Finally, an $f$-step does not receive a special up charge.
\end{proof}

We examine packets scheduled by $\ToggleH$ from left to right, that is
in order of time.  For each time step $t$, if $p$ is the packet scheduled at time $t$, we
want to show that the charge to step $t$ is at most $\phi w_p$.
However, as it turns out, this will not always be true. In one
case we will also consider the next step $t+1$ and the packet $p'$ scheduled in step $t+1$,
and show that the total charge to steps $t$ and $t+1$ is at most $\phi (w_p+w_{p'})$. 

Let $t$ be the current step. We consider several cases.


\smallskip
\noindent
\mycase{1} $t$ is an $e$-step. 
By the definition of $\ToggleH$, $w_e\geq w_h/\phi^2$ and $d_e=t$;
the latter implies that step $t$ receives no regular back charge.
We further note that the heaviest pending packet $h'$ in step $t+1$ is either 
released at time $t+1$ or it coincides with $h$, which is still pending and 
became unmarked by the algorithm in step $t$; in either case $h'$ is unmarked 
at the beginning of step $t+1$, which implies that step $t+1$ is an $f$-step.
Thus, step $t$ receives no special back charge, which, combined with the 
previous observation, implies it receives no back charge of any kind.

Now we claim that the up charge is at most $w_h/\phi$. For a special
or regular up charge this follows from its definition. For a full up
charge, the job $j$ is pending at time $t$ for $\ToggleH$ and $j\neq
h$ (as for $j=h$ the special charges are used). This implies
that $w_j<w_h/\phi$, as otherwise $w_s\geq w_h/\phi$ and $t$ would be
an $f$-step. Thus the full charge is $w_j\leq w_h/\phi$ as well.

Using $w_e\geq w_h/\phi^2$, the charge is at most $w_h/\phi\leq\phi
w_e$ and we are done.


\smallskip
\noindent
\mycase{2} $t$ is an $f$-step and $t$ does not receive a back charge. 
Then $t$ can only receive an up-charge, and this
up charge is at most $w_h\leq \phi w_f$, where the inequality follows from the
definition of $f$.


\smallskip
\noindent
\mycase{3} $t$ is an $f$-step and $t$ receives a special back charge. 
From the definition of special charges, the next step is an $e$-step,
and therefore $h'$ is marked at its beginning.  Since the only
packet that may be marked after an $f$-step is $h$, we thus have
$h=h'=j'$, and the special back charge is $w_h/\phi^2$.
Since $f\prec h$, the adversary cannot schedule $f$ after step $t$, so
step $t$ cannot receive a regular back charge.

We claim that the up charge to step $t$ is at most $w_f$.  Indeed,
a regular up charge is at most $w_h/\phi^2 \leq w_f$, and a special
up charge does not happen in an $f$-step. To show this bound for
a full up charge, assume for contradiction that
$w_j>w_f$. This implies that $j\neq f$ and, since
$\ToggleH$ scheduled $f$, we have $d_j>d_f$. In particular $j$ is
pending at time $t+1$. Thus 
$w_{s'}\ge w_j > w_f \ge w_h/\phi$, contradicting the fact that
$t+1$ is an $e$-step. Therefore the full charge is $w_j\le w_f$, as claimed.

As $w_h\leq \phi w_f$, the total charge to $t$ is at most
$w_f+w_h/\phi^2\leq w_f+w_f/\phi=\phi w_f$. 


\smallskip
\noindent
\mycase{4} $t$ is an $f$-step, $t$ receives a regular back charge and no
special back charge, and
$f=h$. The up charge is at most $w_h/\phi$ by Lemma~\ref{l:up} and the back
charge is at most $w_h$, thus
the total charge is at most $w_h+w_h/\phi=\phi w_h$, and we are done.


\smallskip
\noindent
\mycase{5} $t$ is an $f$-step, $t$ receives a regular back charge and no
special back charge, and
$f\neq h$. Let $\bart$ be the step when the adversary schedules
$f$. We distinguish two sub-cases.


\smallskip
\noindent
\mycase{5.1} In step $\bart$, a packet of weight at least $w_h/\phi$
is pending for the algorithm. 
Then the regular back charge to $t$ is at most 
$w_f-(w_h/\phi)/\phi^2=w_f-w_h/\phi^3$.  As the up charge to $t$ is at most
$w_h/\phi$ by Lemma~\ref{l:up}, the total charge to $t$ is at most
$w_h/\phi+w_f-w_h/\phi^3= w_f+w_h/\phi^2\leq (1+1/\phi) w_f=\phi w_f$,
and we are done.


\smallskip
\noindent
\mycase{5.2} In step $\bart$, no packet of weight at least
$w_h/\phi$ is pending for the algorithm.
In this case we consider the charges to steps $t$ and $t+1$
together. First, we claim the following.

\begin{figure}
\centering
\includegraphics[width=7cm]{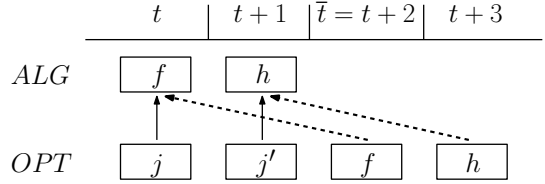}
\caption{An illustration of the situation in Case~5.2. Up charges are denoted by solid arrows
and back charges by dashed arrows.}
\label{fig:ToggleH-case5.2}
\end{figure}

\begin{lemma}\label{lem:aux}
$\ToggleH$ schedules $h$ in step $t+1$. Furthermore, step $t+1$
receives no special charge and it receives an up charge of at most
$w_h/\phi^2$. 
\end{lemma}
\begin{proof}
Since $f\neq h$, we have $f\prec h$ and thus, using also the
definition of $\bart$ and 4-boundedness, $\bart\leq d_f<d_h\leq
t+3$. The case condition implies that $h$ is not pending at $\bart$,
thus $\ToggleH$ schedules $h$ before $\bart$. The only
possibility is that $\ToggleH$ schedules $h$ in step $t+1$,
$\bart=d_f=t+2$, and $d_h=t+3$; see Figure~\ref{fig:ToggleH-case5.2}
for an illustration. This also implies that $\ToggleH$
marks $h$ in step $t$. 

We claim that $w_{s'}<w_h/\phi$. Indeed, otherwise either $s'$ is
pending in step $t+2$, contradicting the condition of Case 5.2, or
$d_{s'}=t+1<d_h$, thus $s'$ is a better candidate for $f'$ than $h$,
which contradicts the fact that the algorithm scheduled
$f'=h$. 

The claim also implies that $h'=h$, as otherwise $w_{s'}\geq
w_h$. Since $h=h'$ is scheduled in step $t+1$, there is no
marked packet in step $t+2$ and $t+2$ is an $f$-step; thus there is
no special back charge to $t+1$. 

We note that step $t+1$ is also an $f$-step, since $\ToggleH$
schedules $h$ in step $t+1$ and $d_h>t+1$.  Since $h'=h$ is marked
when step $t+1$ starts and $w_{s'}< w_h/\phi$, the reason that step
$t+1$ is an $f$-step must be that $d_{e'}>t+1$.

There is no special up charge to step $t+1$ as it is an $f$-step.  If
the up charge to step $t+1$ is a regular up charge, by definition it
is at most $w_{h'}/\phi^2=w_h/\phi^2$ and the lemma holds.

The only remaining case is that of a full up charge to step $t+1$ from 
a packet $j'$ scheduled by the adversary in step $t+1$ and pending for
$\ToggleH$ in step $t+1$.
Since $j'\neq h$, it is a candidate for $s'$, and thus
$w_{j'}<w_h/\phi\leq w_f$.  The earliest-deadline property of the
adversary schedule implies that $j'\prec f$; together with $d_f=t+2$
and $w_{j'}<w_f$ this implies $d_{j'}=t+1$. Therefore
$w_{j'}<w_h/\phi^2$, as otherwise $j'$ is a candidate for $e'$, but we
have shown that $d_{e'}>t+1$. Thus the regular up charge is at most
$w_{j'}<w_h/\phi^2$ and the lemma holds also in the remaining case.
\end{proof}

By Lemma~\ref{lem:aux}, step $t+1$ receives no special charge and an
up charge of at most $w_h/\phi^2$ and $\ToggleH$ schedules $h$ in step
$t+1$. Step $t+1$ thus also receives a regular back charge of at most
$w_h$.  So the total charge to step $t+1$ is at most $w_h/\phi^2 + w_h
\le w_f/\phi + w_h$.  Moreover, using Lemma~\ref{l:up}, the total charge
to step $t$ is at most $w_h/\phi + w_f$.  Thus, the total charge to
these two steps is at most $(w_h/\phi + w_f) + (w_f/\phi + w_h) =
\phi(w_f + w_h)$, as $f$ and $h$ are the two packets scheduled by
$\ToggleH$.


\smallskip

In each case we have shown that a step or a pair of 
consecutive steps receive a total charge of at most $\phi$
times the weight of packets scheduled in these steps.
Thus $\ToggleH$ is $\phi$-competitive for the $4$-bounded case.
\end{proof}
%


\section{An Algorithm for 2-Bounded Instances with Lookahead}\label{sec:lookaheadalgo}



In this section, we present an algorithm for \emph{$2$-bounded} {\BDPS} \emph{with 1-lookahead}, 
as defined in Section~\ref{sec:definitions}.

Consider some online algorithm $\calA$. Recall that,
for a time step $t$, packets \emph{pending} for $\calA$ are those that are released at or
before time $t$ and have neither expired nor been scheduled by $\calA$ before time $t$.
\emph{Lookahead} packets at time $t$ are the packets with release time $t+1$.

For $\calA$, we define the \emph{plan} in step $t$ to be the optimal schedule
in the time interval $[t,\infty)$ that consists
of pending and lookahead packets at time $t$ and has the
earliest-deadline property.
For $2$-bounded instances, this plan will only use slots $t$, $t+1$ and $t+2$.
We will typically denote the packets in the plan scheduled
in these slots by $p_1,p_2,p_3$, respectively. The earliest-deadline
property then implies that if both $p_1$ and $p_2$ have release time $t$
and deadline $t+1$ then $p_1$ is heavier than $p_2$ and similarly for
$p_2$ and $p_3$.

\goodbreak

\myparagraph{Algorithm \LCalpha.}
Fix some parameter $\alpha > 1$. At any time step $t$, the algorithm proceeds as follows:
\begin{tabbing}
aaa \= aaa \= aaa \= aaa \= aaa \= aaa \= \kill
\> let $p_1,p_2,p_3$ be the plan at time $t$
\\
\> \textbf{if} $r_{p_2} = t$ \textbf{and} 
		$w_{p_1} <  \min(\, w_{p_2} \,,\,  w_{p_3} \,,\, \frac{1}{2\alpha} (w_{p_2} + w_{p_3}) \,)$
\\
\>\> \textbf{then} schedule $p_2$
\\
\> \textbf{else} schedule $p_1$
\end{tabbing}

Note that if the algorithm schedules $p_2$ then $p_1$ must be expiring,
for otherwise $w_{p_1} > w_{p_2}$ (by canonical ordering).
Also, the scheduled packet is at least as heavy
as the heaviest expiring packet $q$,
since clearly $w_{p_1}\ge w_q$ and the algorithm schedules $p_2$ only
if $w_{p_1}<w_{p_2}$.


\myparagraph{Analysis.}
We set the parameter $\alpha$ and constants $\delta$ and $R$
which we will use in the analysis so that they satisfy the following equalities:
\begin{align}
2-\delta - \frac{R + 2\delta - 1}{\alpha} = R \label{eq:forwardCh} \\ 
1 - 2\delta + 2\alpha\delta = R \label{eq:chainCharges} \\
1 + \frac{1}{2\alpha} = R \label{eq:splitCharges}
\end{align}
By solving these equations we get
$\alpha=\onefourth (\sqrt{13} + 3) \approx 1.651$, 
$\delta = \onesixth (5 - \sqrt{13}) \approx 0.232$,
and $R = \half(\sqrt{13} - 1) \approx 1.303$.

In this section we will prove the following theorem:

\begin{theorem}
The algorithm \LCalpha is $R$-competitive for packet scheduling
on 2-bounded instances for $R = \half(\sqrt{13} - 1) \approx 1.303$ if $\alpha = \onefourth(\sqrt{13} + 3) \approx 1.651$.
\end{theorem}

We also use the following properties of these constants:
\begin{align}
2-R-3\delta = 0\label{eq:2-R-3delta}\\
2-R-2\delta > 0 \label{ineq:2-R-2delta}\\
1-\delta-\frac{R-1+2\delta}{2\alpha} > 0 \label{ineq:1-d-frac}\\
1-\frac{R}{2\alpha} > 0\label{ineq:1-R/2alpha}\\
3\alpha\delta < R \label{ineq:chainBegCh}\\
2-\frac{R}{\alpha} < R\label{ineq:fwdChFromSingleton}
\end{align}
where~(\ref{eq:2-R-3delta}) follows from~(\ref{eq:forwardCh})
and~(\ref{eq:chainCharges}) and strict inequalities can be verified
numerically.

Let \ALG be the schedule produced by \LC. 
Let us consider an optimal schedule \OPT (a.k.a.\ schedule of the adversary) satisfying the 
canonical ordering, i.e., if a packet $x$ is scheduled before a packet $y$ in {\OPT}
then either $y$ is released after $x$ is scheduled or $x\prec y$.
Recall that we are assuming w.l.o.g.\ that the weights of packets are different.

The analysis of \LC is based on a charging scheme.
First we define a few packets by their schedule times; see Figure~\ref{fig:packetdef}.

\begin{figure}[ht]
	\begin{minipage}{3in}
	\begin{compactitem}
	\item $i$ = packet scheduled in step $t-1$ in \OPT,
	\item $j$ = packet scheduled in step $t$ in \OPT,
	\item $k$ = packet scheduled in step $t+1$ in \OPT,
	\item $e$ = packet scheduled in step $t-1$ in \ALG,
	\item $f$ = packet scheduled in step $t$ in \ALG,
	\item $g$ = packet scheduled in step $t+1$ in \ALG,
	\item $h$ = packet scheduled in step $t+2$ in \ALG.
	\end{compactitem}
	\end{minipage}
	%
	%
	\begin{minipage}{2.3in}
	\includegraphics[width=2.3in]{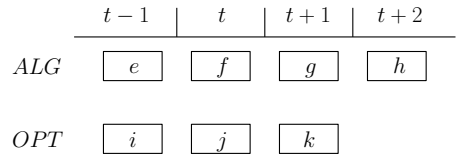}
	\end{minipage}
	\caption{Packet definition.}
	\label{fig:packetdef}
\end{figure}

\smallskip

\myparagraph{Informal description of charging.}
We use three types of charges.  The adversary's packet $j$ in step
$t$ is charged using a \textit{full charge} either to step $t-1$ if
\ALG schedules $j$ in step $t-1$ or to step $t$ if $w_f\geq w_j$
(including the case $f=j$) and $f$ is not in step $t+1$ in \OPT; the
last condition assures that step $t$ does not receive two full
charges.

The second type are \textit{split charges} that occur in step $t$ if
$w_f > w_j$, $j$ is pending in step $t$ in {\ALG} and $f$ is in step
$t+1$ in \OPT, i.e., step $t$ receives a full back charge from $f$.
In this case, we distribute the charge from $j$ to $f$ and another
relatively large packet $f'$ scheduled in step $t+1$ or $t+2$ in
\ALG; we shall prove that one of these steps satisfies $2\alpha\scdot
w_j<w_f+w_f'$. We charge to step $t+2$ only when it is necessary,
which allows us to prove that split-charge pairs are pairwise
disjoint. Also, in this case we analyze the charges to both steps
together, thus it is not necessary to fix a distribution of the weight
to the two steps.

The remaining case is when $w_f < w_j$ and $j$ is not scheduled in
$t-1$ in \ALG.  We analyze these steps in maximal consecutive
intervals, called \textit{chains} and the corresponding charges are
\textit{chain charges}.
Inside each chain we distribute the charge of each packet $j$
scheduled at $t$ in \OPT to steps $t-1$, $t$ and $t+1$, if these steps
are also in the chain.  The distribution of weights shall depend on
a parameter $\delta$.  Packets at the beginning and at the end of the
chain are charged in a way that minimizes the charge to steps outside
of the chain. In particular, the step before a chain receives no charge
from the chain. 


\myparagraph{Notations and the charging scheme.} 
A step $t$ for which $w_f < w_j$ and $j$ is pending in step $t$ in \ALG
is called \textit{a chaining step}.
A maximal sequence of successive chaining steps is called a \textit{chain}.
The chains with a single step are called \textit{singleton chains},
the chains with at least two steps are called \textit{long chains}.

The pair of steps that receive a split charge from the same packet is
called a \textit{split-charge pair}. The charging scheme does not
specify the distribution of the weight to the two steps of the
split-charge pair, as the charges to them are analyzed together.

Packet $j$ scheduled in \OPT at time $t$ is charged
according to the first rule below that applies. See
Figure~\ref{fig:nonchaining} for an illustration of the first four
(non-chaining) charges and Figure~\ref{fig:chaining} for an illustration
of the chaining charges.

\begin{figure} 
\centering
\includegraphics[scale=0.28]{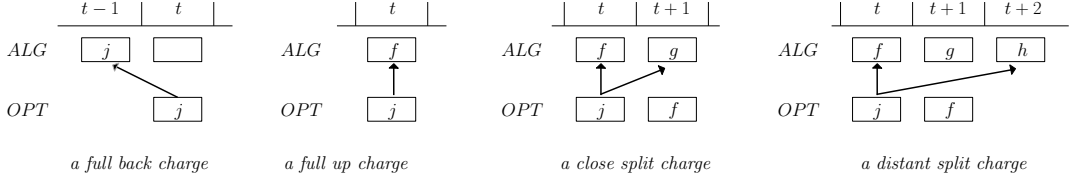}
\caption{Non-chaining charges. Note that for split charges $f$ is
  scheduled in step $t+1$ in \OPT 
which follows from the fact that we do not charge $j$ using a full up charge.}
\label{fig:nonchaining}
\end{figure}

\begin{figure} 
\centering
\includegraphics[scale=0.35]{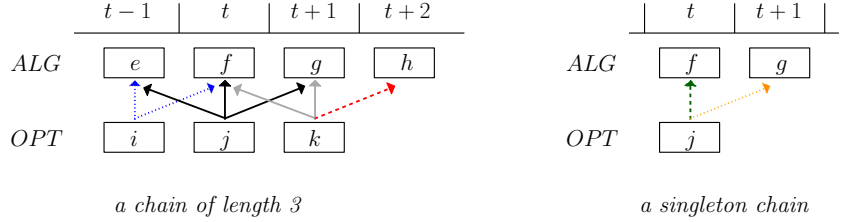}

\caption{On the left, a chain of length 3 starting in step $t-1$ and
ending in step $t+1$. The \emph{chain beginning charges} are denoted
by dotted (blue) lines, the chain end charges are denoted by
gray lines and the \emph{forward charge from a chain} is depicted by a dashed (red) arrow.
Black arrows denote the \emph{chain link charges}.
On the right, an example of a singleton chain, with the
\emph{up charge from a singleton chain} denoted with a dashed (green)
line and the \emph{forward charge from a singleton chain} denoted with
a dotted (orange) line.}

\label{fig:chaining}
\end{figure}

\begin{enumerate}
\item \label{ch:back} If $j$ is scheduled in step $t-1$ in \ALG (that is, $e=j$),
  charge $w_j$ to step $t-1$. We call this charge a \textit{full back charge}.
\item \label{ch:up} If $w_f\geq w_j$ and $f$ is not scheduled in step $t+1$ in \OPT (in particular, if $j=f$),
charge $w_j$ to step $t$. We call this charge a \textit{full up charge}.
\item \label{ch:split} If $w_f > w_j$ and at least one of the following holds:
\begin{compactitem}
\item $2\alpha\scdot w_j < w_f + w_g$,
\item $g$ does not get a full back charge and $2\alpha\scdot (w_{p_1} - w_g) < w_f + w_g$ where
$p_1$ is the first packet in the plan at time $t$,
\end{compactitem}
then charge $w_j$ to the pair of steps $t$ and $t+1$. We call this charge a \textit{close split charge}.
\item \label{ch:distant} If $w_f > w_j$, then charge $w_j$ to the pair
  of steps $t$ and $t+2$. We call this charge a \textit{distant split charge}.
\item \label{ch:chain} Otherwise step $t$ is a chaining step, as $w_f
  < w_j$ and \ALG\ does not schedule $f$ in step $t-1$ by the previous
  cases. We distinguish the following subcases.
\begin{enumerate}
\item \label{ch:chainSingle} If step $t$ is (the only step of) a
  singleton chain, then charge $\min(w_j, R\scdot w_f)$ to step $t$ and
  $w_j - R\scdot w_f$ to step $t+1$ if $w_j>R\scdot w_f$. We call these
  charges an \textit{up charge from a singleton chain} and a
  \textit{forward charge from a singleton chain}.
\item \label{ch:chainBeginning} If step $t$ is the first step of a
  long chain, charge $2\delta\scdot w_j$ to step $t$, and
  $(1-2\delta)\scdot w_j$ to step $t+1$. We call these charges
  \textit{chain beginning charges}.
\item \label{ch:chainEnd} If step $t$ is the last step of a long
  chain, charge $\delta\scdot w_j$ to step $t-1$, $(R - 1 +
  2\delta)\scdot w_f$ to step $t$, and $(1-\delta)\scdot w_j - (R - 1 +
  2\delta)\scdot w_f$ to step $t+1$. We call these charges
  \textit{chain end charges}; the charge to step $t+1$ is called a
  \textit{forward charge from a chain}.  (Note that we always have
  $(1-\delta)\scdot w_j > (R - 1 + 2\delta)\scdot w_f$, since $w_j >
  w_f$ and $1-\delta = R - 1 + 2\delta$ which follows
  from~(\ref{eq:2-R-3delta}).)
\item \label{ch:chainInside} Otherwise, i.e., step $t$ is inside a
  long chain, charge $\delta\scdot w_j$ to step $t-1$, $\delta\scdot
  w_j$ to step $t$, and $(1-2\delta)\scdot w_j$ to step $t+1$. We call
  these charges \textit{chain link charges}.
\end{enumerate}
\end{enumerate}

To estimate the competitive ratio we need to show that each step
or a pair of steps does not receive too much charge.
We start with a useful observation about plans of Algorithm~{\LCalpha}, 
that will be used multiple times in our proofs.

\begin{lemma}\label{l:scheduledTightPacketNotTooLight}
Consider a time $t$, where the algorithm has two pending packets $a$, $b$ and a
lookahead packet $c$ with the following properties:
$d_a =t$, $(r_b,d_b) =  (t, t+1)$, $(r_c,d_c) =  (t+1, t+2)$, 
and $w_a < \min (w_b, w_c)$.
If the algorithm schedules packet $a$ in step $t$ then the plan
at time $t$ is $a,b,c$, and $2\alpha\scdot w_a\ge w_b+w_c$.
\end{lemma}

\begin{proof}
We claim that there is no pending or lookahead packet $q
\notin\braced{b,c}$ heavier than $a$. Suppose for a contradiction that
such a $q$ exists. Then a schedule containing packets $q,b,c$ in some
order is feasible and has larger profit than $a,b,c$. This implies
that the plan does not contain $a$ and thus $a$ cannot be scheduled,
contradicting the assumption of the lemma.

The schedule $a,b,c$ is feasible and the claim above implies that it
is optimal, thus it is the plan. It remains to show that $2\alpha\scdot
w_a \geq w_b+w_c$, which follows easily by a contradiction: Otherwise 
$2\alpha\scdot w_a < w_b+w_c$ and \LCalpha\ would schedule $b$,
contradicting the assumption of the lemma.
\end{proof}

Next, we will provide an analysis of full, split and chain charges,
starting with full and split charges. We prove several lemmas from which the analysis follows.
We fix some time slot $t$, and use the notation from Figure~\ref{fig:packetdef} for
packets at time slots $t-1$, $t$, $t+1$ and $t+2$ in the schedule of the algorithm and
the optimal schedule.


\smallskip

\myparagraph{Analysis of full charges.}  Using 
Rules~\ref{ch:back} and~\ref{ch:up}, if step $t$
receives a full back charge, then the condition of Rule~\ref{ch:up} guarantees
that it will not receive a full up charge. This gives us the following
observation.

\begin{lemma}\label{l:oneFullCharge}
Step $t$ receives at most one full charge, i.e., a charge by Rule \ref{ch:back}
or \ref{ch:up}.
\end{lemma}


\myparagraph{Analysis of split charges.}
We now analyze close and distant split charges. The crucial property of split charges
is that, similar to full charges, each step receives at most one split charge.
Before we prove this, we establish several useful properties of split charges.


\begin{lemma}\label{obs:splitCharges}
Let the plan at time $t$ be $p_1,p_2,p_3$.
If $j$ is charged using a close or a distant split charge, then the following holds:
\begin{compactenum}[(a)]
\item $j$ is not scheduled by the algorithm in step $t-1$, i.e., $j$ is pending for the algorithm in step $t$.
\item $d_f = t+1$ and $f$ is scheduled in step $t+1$ in \OPT (that is, $k=f$). 
In particular, step $t$ receives a full back charge.
\item $d_j = t$ and $w_j\le w_{p_1}$.
\item $p_2 = f$.
\end{compactenum}
\end{lemma}

\begin{proof}
By Rule~\ref{ch:back}, packet $j$ would be charged using a full back
charge if it were scheduled in step $t-1$, implying (a). The case
conditions for split charges in the charging scheme imply that \OPT
schedules $f$ in step $t+1$ and $w_f > w_j$. Now (b) follows from the
fact that we do not charge $j$ using a full up charge.

To show (c), note that if $j$ is not expiring, then $j$ and $f$ would
have equal deadlines.  As we also have $w_f > w_j$, $f$ would be
scheduled before $j$ in \OPT by the canonical ordering, a
contradiction.  The inequality $w_j\le w_{p_1}$ now follows from the
definition of the plan.

It remains to prove (d). Towards contradiction, suppose that $f=p_1$.
We know that $j$ is expiring and thus it is not in the plan.
If $d_{p_2} = t+1$ then the optimality of the plan implies
$w_{p_2} > w_j$ (otherwise $j, f, p_3$ would be a better plan),
so, since $p_2$ is not in {\OPT},
we could improve {\OPT} by scheduling $f$ in step $t$ and $p_2$ in step $t+1$.

Next, assume that $d_{p_2} = t+2$.
The optimality of the plan implies that 
$w_{p_2} > w_j$ and $w_{p_3} > w_j$.
Since both $p_2$, $p_3$ have deadline $t+2$,
at least one of them is not scheduled in \OPT.
So \OPT could be improved by scheduling $f$ in step $t$ and one of $p_2$ or $p_3$ in step $t+1$.
In both cases we get a contradiction with the optimality of \OPT.
\end{proof}

We show a useful lemma about a distant split charge from which we derive an upper bound
on $w_j$, similar as the upper bound in the definition of close split charge.

\begin{lemma}\label{l:distantSplitCh:g<p3}
If $j$ is charged using a distant split charge, then
$w_g < w_{p_3}$ where $p_3$ is the third packet in the plan at time $t$, and $d_g = t+1$.
\end{lemma}

\begin{proof}
Suppose that $w_g \ge w_{p_3}$.
Then, from Lemma~\ref{obs:splitCharges}(c),(d) and the choice of $p_2 = f$ in the algorithm, we have that
$2\alpha w_j \le 2\alpha w_{p_1}
				< w_{p_2} + w_{p_3}
				\le w_f + w_g$,
so we would use the close split charge in step $t$, not the distant one.
Thus $w_g < w_{p_3}$, as claimed.

To prove the second part, if we had $d_g = t+2$ then, since the
algorithm chose $g$ in step $t+1$ and also $d_{p_3} = t+2$,
we would also have that $w_g\ge w_{p_3}$ -- a contradiction.
\end{proof}


\begin{lemma}\label{l:distributeCharge}
If $j$ is charged using a distant split charge then $2\alpha\scdot w_j < w_f + w_h$. 
(Recall that $h$ is the packet scheduled in step $t+2$ in \ALG.)
\end{lemma}

\begin{proof}
Let $p_1, p_2, p_3$ be the plan in step $t$.
By Lemma~\ref{obs:splitCharges}(d) we have that $f=p_2$. 
Thus $2\alpha\scdot w_{p_1} < w_{p_2} + w_{p_3}$ by the definition of the algorithm.
By Lemma~\ref{obs:splitCharges}(c), $j$ is expiring and $w_j \le w_{p_1}$.
As $g\neq p_3$ by Lemma~\ref{l:distantSplitCh:g<p3},
the algorithm has $p_3$ pending in step $t+2$ where it is expiring, implying that
$w_{p_3}\le w_h$. 
Putting it all together, we get
$2\alpha w_j \le 2\alpha w_{p_1} < w_{p_2} + w_{p_3} \le w_f + w_h$.
\end{proof}

For a split charge from $j$ in step $t$, let $t'$ be the other step
that receives the split charge from $j$; that is, $t'=t+1$ for a close split
charge and $t'=t+2$ for a distant split charge.
We now show that split-charge pairs are pairwise disjoint.


\begin{lemma}\label{l:noTwoDistribCharges}
If $j$ is charged using a split charge to a pair of steps $t$ and $t'$,
then neither of $t$ and $t'$ is involved in another pair that receives
a split charge from a packet $j'\neq j$.
\end{lemma}

\begin{proof}
No matter which split charge we use for $j$,
using Lemma~\ref{obs:splitCharges}(b), step $t+1$ does not receive a
split charge from $k=f$. By a similar argument, since $j$ is not scheduled
in step $t-1$ in {\ALG}, step $t$ does not receive 
a close split charge from the packet scheduled in step $t-1$ in {\OPT}.

It remains to prove that if $j$ is charged using a distant split charge,
then the packet $\ell$ scheduled in step $t+2$ in \OPT is not charged using
a split charge. (This also ensures that step $t$ does not receive a
distant split charge 
from a packet scheduled in step $t-2$ in \OPT.)

For a contradiction, suppose that packet $\ell$ is charged using a split charge.
Let $p_1, p_2, p_3$ be the plan in step $t$.
Recall that $g$ and $h$ are the packets scheduled in steps $t+1$ and $t+2$ in \ALG.

From Lemma~\ref{l:distantSplitCh:g<p3}, step $t+1$ does not receive a full back charge. Since we did not apply
the close split charge for $j$ in Rule~\ref{ch:split}, we must have
\begin{equation}\label{eq:no2splitChEq}
2\alpha (w_{p_1} - w_g) \ge w_f+w_g \ge w_f.
\end{equation}

By Lemma~\ref{obs:splitCharges}(b) applied to step $t+2$, we get $d_h = t+3$.
Since $d_{p_3} = t+2$, we get $w_{p_3} < w_h$.
We now use Lemma~\ref{l:scheduledTightPacketNotTooLight} for step $t+1$ with
$a=g, b=p_3$, and $c=h$. We note that all the assumptions of the lemma are
satisfied: we have $d_g = t+1$, $(r_{p_3},d_{p_3}) = (t+1,t+2)$, 
$(r_h,d_h) = (t+2,t+3)$, and  $w_g < w_{p_3} < w_h$.
This gives us that 
$2\alpha  w_g \ge w_{p_3}+w_h > w_{p_3}$.

Since the algorithm schedules $f=p_2$ in step $t$, we have
$2\alpha  w_{p_1}  < w_f+w_{p_3}$.
Subtracting the inequality derived in the previous paragraph, we
get $2\alpha  (w_{p_1}-w_g)  < (w_f + w_{p_3}) - w_{p_3}  = w_f$
-- a contradiction with (\ref{eq:no2splitChEq}).
This completes the proof.
\end{proof}

The lemmas above allow us to estimate the total of full and split charges. 

\begin{lemma}\label{l:splitChargeUB}
If $j$ is charged using a split charge to a pair of steps $t$ and $t'$,
then the total of full and split charges
to steps $t$ and $t'$ does not exceed $R\scdot (w_f + w_{f'})$
where $f'$ is the packet scheduled in step $t'$ in \ALG.
\end{lemma}

\begin{proof}
Each of steps $t$ and $t'$ may receive a full charge,
but each step at most one full charge from a packet of smaller or equal weight
by Lemma~\ref{l:oneFullCharge} and charging rules.

If we use a distant split charge or if step $t'$ gets a full back charge, then
$2\alpha\scdot w_j < w_f + w_{f'}$ by Lemma~\ref{l:distributeCharge} or Rule~\ref{ch:split}.
Thus the total of full and split charges to steps $t$ and $t'$
is upper bounded by 
\begin{equation*}
w_f +  w_{f'} + w_j < w_f +  w_{f'} + \frac{w_f +  w_{f'}}{2\alpha}
= \left(1 + \frac{1}{2\alpha}\right)\scdot (w_f +  w_{f'})= R\scdot (w_f +  w_{f'})
\end{equation*}
where we used~(\ref{eq:splitCharges}) in the last step.

Otherwise, i.e., if we use a close split charge and step $t'=t+1$ does
not get a full back charge, 
then we have $2\alpha\scdot (w_{p_1} - w_{f'}) < w_f + w_{f'}$ by Rule~\ref{ch:split}.
Since $d_j = t$ by Lemma~\ref{obs:splitCharges}(c), we have $w_j\le w_{p_1}$
and $2\alpha\scdot (w_j - w_{f'}) < w_f + w_{f'}$.
Also, step $t+1$ does not receive a full up charge by Lemma~\ref{obs:splitCharges}(b).
We thus bound the total of full and split charges to steps $t$ and $t+1$ by
\begin{equation*}
w_f +  w_j < w_f +  \frac{w_f +  (2\alpha + 1)\scdot w_{f'}}{2\alpha}
= \left(1 + \frac{1}{2\alpha}\right)\scdot (w_f +  w_{f'}) = R\scdot (w_f +  w_{f'})
\end{equation*}
using~(\ref{eq:splitCharges}) in the last step again.
\end{proof}
 


\myparagraph{Analysis of chain charges.}
We now analyze chaining steps starting with a lemma below consisting of several useful observations.
In particular, Part~(c) motivates the name ``chaining'' for such steps.


\begin{lemma}\label{obs:chainSteps}
If step $t$ is a chaining step, then the following holds:
\begin{compactenum}[(a)]
\item $d_j = t+1$,
\item $d_f = t$.
\end{compactenum}
Moreover, if step $t+1$ is also a chaining step, then
\begin{compactenum}
\item[(c)] $j$ is scheduled by the algorithm in step $t+1$, i.e., $g=j$,
\item[(d)] $2\alpha\scdot w_f \ge w_j+w_k$ (recall that $k$ is the packet
scheduled in step $t+1$ in \OPT).
\end{compactenum}
\end{lemma}

\begin{proof}
Recall that Algorithm~\LCalpha never schedules a packet lighter than
the heaviest expiring packet.  As in step $t$ it schedules $f$ with
$w_f < w_j$ (by Rule~\ref{ch:chain} for chain charges), (a) follows.
Furthermore, it follows that $f$ is expiring in step $t$,
because otherwise the algorithm would schedule $j$, since both would have the same deadline and $j$ is heavier.
Thus (b) holds as well.

Now assume that step $t+1$ is also in the chain and
for a contradiction suppose that $g\neq j$.
Since $j$ is expiring and pending for the algorithm in step $t+1$,
we have $w_g > w_j$ and $w_k > w_g$ as step $t+1$ is in the chain.

Summarizing, the algorithm sees all packets $f,j,g,k$ in step $t$ (some are pending and some may be lookahead packets), 
and they are all distinct packets with $w_f < w_j < w_g < w_k$, $d_f = t$, $(r_j,d_j) = (t,t+1)$, and both
$g$ and $k$ can be feasibly scheduled at time $t+1$. Thus, independently of the release times and deadlines of
$g$ and $k$, the plan at time $t$ containing $f$ would not be optimal -- a contradiction.
This proves that (c) holds.

Finally, we show (d).
Since $f$ is expiring in step $t$ by (b) and both $j$ and $k$ are considered for the plan at time $t$
and satisfy $(r_j, d_j) = (t, t+1)$, $(r_k, d_k) = (t+1, t+2)$, $w_f < w_j < w_k$,
we use Lemma~\ref{l:scheduledTightPacketNotTooLight} with $a=f, b=j,$ and $c=k$
and get the inequality in (d).
\end{proof}

First we show that chaining steps does not receive charges of other types.


\begin{lemma}\label{l:noFullorSplitChargeInAChain}
If step $t$ is a chanining step, then $t$ does not receive a full charge or a split charge.
\end{lemma}

\begin{proof}
By Lemma~\ref{obs:chainSteps}, $f$ is expiring, thus step $t$ does not receive a full back charge.
As $w_j>w_f$, the step also does not get a full up charge or a split charge from step $t$.
So it remains to show that $f$ does not receive a split charge.

First observe that step $t$ cannot receive a close split charge from step $t-1$
in \OPT, because $j$ is pending in step $t$ in \ALG, while
Lemma~\ref{obs:splitCharges}(b) states that a split charge 
from step $t-1$  would require $j$ to be scheduled at time $t-1$ in \ALG.

Finally, we show that step $t$ does not receive a distant split charge.
For a contradiction, suppose that step $t$ receive a distant split charge from the packet $x$
scheduled in step $t-2$ in \OPT. Let $p_1, p_2, p_3$ be the plan in step $t-2$. 
According to Lemma~\ref{obs:splitCharges}(d) and (b), $p_2$ is scheduled
in step $t-2$ in \ALG and in step $t-1$ in \OPT.
Moreover, by Lemma~\ref{obs:splitCharges}(c),
$x$ is pending and expiring in step $t-2$ and $w_{p_1}\ge w_x$.
As the algorithm scheduled $p_2$ in step $t-2$ we get $r_{p_2} = t-2$ and $w_{p_1} < w_{p_3}$. 

Observe that $p_3$ is not scheduled in \OPT, since it is expiring in step $t$
and $j$ is not expiring, by Lemma~\ref{obs:chainSteps}(a). Thus we could increase the weight of \OPT
if we scheduled $p_2$ in step $t-2$ instead of $x$ and $p_3$ in step $t-1$.
This contradicts the optimality of \OPT.
\end{proof}


We now analyze how much charge does each chaining step get.

\begin{lemma}\label{l:chargingToChain}
If step $t$ is a chaining step,
then it receives a charge of at most $R\scdot w_f$.
\end{lemma}

\begin{proof}
By Lemma~\ref{l:noFullorSplitChargeInAChain}, step $t$ does not receive any full or 
split charges; therefore we just need to prove that the total of chain charges
to step $t$ does not exceed $R\scdot w_f$.

\smallskip
\noindent
\mycase{1} $t$ is the last step of a chain.
If $t$ is the only step in the chain then Rule~\ref{ch:chainSingle} 
implies directly that the charge to $t$ is at most $R\scdot w_f$. 
Otherwise, Lemma~\ref{obs:chainSteps}(c) implies that $f$ is scheduled in step $t-1$ in \OPT, and thus
the charge from step $t-1$ is $(1-2\delta)\scdot w_f$. The charge from step $t$
is at most $(R-1+2\delta)\scdot w_f$ by Rule~\ref{ch:chainEnd}.
So the total charge is at most $R\scdot w_f$.

\smallskip
\noindent
\mycase{2} $t$ is not the last step of a chain.
Since step $t+1$ is also in the chain,
by Lemma~\ref{obs:chainSteps}(c) we have that $j$ is scheduled in step $t+1$ in \ALG
and \OPT has a packet $k$ with $w_k>w_j$ in step $t+1$.
From Lemma~\ref{obs:chainSteps}(d) we know that $2\alpha\scdot w_f \ge w_j+w_k$.

There are two sub-cases. If $t$ is the first step of the chain, then the charge to $t$
is at most
\begin{equation*}
2\delta\scdot w_j + \delta\scdot w_k \le \threehalves \delta\scdot (w_j+w_k)
					\le 3\alpha\delta\scdot w_f 
					< R\scdot w_f,
\end{equation*}
where the last inequality follows from~(\ref{ineq:chainBegCh}).
Otherwise, using Lemma~\ref{obs:chainSteps}(c), $f$ is scheduled in step $t-1$ in \OPT,
so the total charge to step $t$ is at most
\begin{equation*}
(1-2\delta)\scdot w_f + \delta\scdot w_j + \delta\scdot w_k \le
  (1-2\delta)\scdot w_f + 2\alpha\delta\scdot w_f = R\scdot w_f
\end{equation*}
where the last equality follows from~(\ref{eq:chainCharges}).
\end{proof}


\myparagraph{Analysis of forward charges from chains.}
We now show that a forward charge from a chain does not cause
an overload on the step just after the chain
which may also get both a full charge and a split charge.
(This is the only case when a step receives charges of all three types.)

For the following lemmas we assume that step $t-1$ is a chaining
step. Recall that $i$ is 
the packet scheduled in step $t-1$ in \OPT and $e$ is the packet
scheduled in step $t-1$ in \ALG. First, we prove some useful observations.

\begin{lemma}\label{l:fwdChargeFromChain}
If step $t$ receives a forward charge from a chain, then the following holds
\begin{compactenum}[(a)]
\item $j\neq e$ (that is, $j$ is not charged using a full back charge),
\item $w_f\geq w_j$,
\item $w_f \geq w_i$.
\end{compactenum}
Moreover, if step $t$ is not in a split-charge pair:
\begin{compactenum}
\item[(d)] $j$ is charged using a full up charge to step $t$,
\item[(e)] step $t$ does not receive a back charge.
\end{compactenum}
\end{lemma}

\begin{proof}
Part~(a) holds because $e$ is expiring in step $t-1$, by Lemma~\ref{obs:chainSteps}(b).
Part~(b) follows from (a) and the fact that step $t$ is not chaining.

To show~(c), Lemma~\ref{obs:chainSteps}(a) implies that
$d_i = t$. Also, since $i\neq e$, $i$ is pending in \ALG in step $t$.
Now (c) follows, because each packet scheduled by the algorithm is at least
as heavy as the the heaviest expiring packet. 

Part~(d) follows from (a) and (b) and the assumption that $t$ is not in a split-charge pair.
Part~(e) follows from (d) and Lemma~\ref{l:oneFullCharge}.
\end{proof}

Note that $f$ may be the same packet as $i$ or $j$. We start with the case in which $f$ is not
in a split-charge pair.

\begin{lemma}\label{l:fwdChFromChain-NoSplitCh}
If step $t$ receives a forward charge from a chain $C$
and $t$ is not in a split-charge pair,
then the total charge to step $t$ is at most $R\scdot w_f$.
\end{lemma}

\begin{proof}
The proof is by case analysis, depending on the relative weights of $j$ and $e$, and
on whether $C$ is a singleton or a long chain.
In all cases we use Lemma~\ref{l:fwdChargeFromChain} and the charging rules
to show upper bounds on the total charge.

\smallskip
\noindent
\mycase{1} $w_j < w_e$. 

\smallskip
\noindent
\mycase{1.1} The chain $C$ is long.
The charge to step $t$ is then at most
\begin{align}
w_j + (1-\delta)\scdot w_i - (R-1+2\delta)\scdot w_e 
	&< w_j + (1-\delta)\scdot w_i - (R-1+2\delta)\scdot w_j 
	\nonumber
	\\
	&= (2-R-2\delta)\scdot w_j + (1-\delta)\scdot w_i
	\nonumber
	\\
	&\le (2-R-2\delta)\scdot w_f + (1-\delta)\scdot w_f 
	\label{eqn: fwdChFromChain-NoSplitCh 1.1}
	\\
	&=	(3-R-3\delta)\scdot w_f = w_f\,.
	\label{eqn: fwdChFromChain-NoSplitCh 1.2}
\end{align}
To justify inequality~(\ref{eqn: fwdChFromChain-NoSplitCh 1.1}), note that
$2-R-2\delta\ge 0$ by~(\ref{ineq:2-R-2delta}) and
$1-\delta\ge 0$ by the choice of $\delta$, so 
we can apply inequalities $w_j\le w_f$ and
$w_i\le w_f$ from Lemma~\ref{l:fwdChargeFromChain}(b) and (c).
The last step~(\ref{eqn: fwdChFromChain-NoSplitCh 1.2}) follows from equation~(\ref{eq:2-R-3delta}).

\smallskip
\noindent
\mycase{1.2} The chain $C$ is singleton.
We assume that $w_i > R\scdot w_e$,
otherwise there is no forward charge from the chain.
Then the charge to step $t$ is
\begin{equation*}
	w_j + w_i - R\scdot w_e 
		\le	w_j + w_i - R\scdot w_j 
 		\le w_i 
		\le w_f\,,
\end{equation*}
where in the last step we used Lemma~\ref{l:fwdChargeFromChain}(c).

\smallskip
\noindent
\mycase{2}
$w_j > w_e$.
We claim first that $j$ is not expiring in step $t$, that is $d_j = t+1$.
Indeed, if we had $d_j = t$, then in step $t-1$ the algorithm would have pending
packets $e$ and $i$, plus packet $j$ (pending or lookahead), that need to be
scheduled in slots $t-1$ and $t$. Since $w_e < w_i$ (because step
$t-1$ is a chaining step) and $w_e < w_j$ (by the case assumption),  
packet $e$ could not be in the plan in step $t-1$ which is a contradiction.
Thus $d_j = t+1$.

Recall that $e$ is expiring in step $t-1$ by Lemma~\ref{obs:chainSteps}(b)
and both $i$ and $j$ are considered for the plan in step $t-1$.
Moreover, we know that $w_i > w_e$, $w_j > w_e$, $(r_i, d_i) = (t-1, t)$ (by Lemma~\ref{obs:chainSteps}(a)),
and $(r_j, d_j) = (t, t+1)$.
We thus use Lemma~\ref{l:scheduledTightPacketNotTooLight} for step $t-1$
with $a=e, b=i,$ and $c=j$, to get that $2\alpha\scdot w_e \ge w_i+w_j$.

\smallskip
\noindent
\mycase{2.1} The chain $C$ is long.
The charge to step $t$ is
\begin{align}
w_j + (1-\delta)\scdot w_i &- (R-1+2\delta)\scdot w_e 
	\nonumber
	\\
	&\le  w_j + (1-\delta)\scdot w_i - (R-1+2\delta)\scdot \frac{w_i + w_j}{2\alpha} 
	\nonumber
	\\
	&= \left(1-\frac{R-1+2\delta}{2\alpha}\right) \scdot w_j + \left(1-\delta-\frac{R-1+2\delta}{2\alpha}\right) \scdot w_i
	\nonumber
	\\
	&\le \left(1-\frac{R-1+2\delta}{2\alpha}\right) \scdot w_f + \left(1-\delta-\frac{R-1+2\delta}{2\alpha}\right) \scdot w_f 
	\label{eqn: fwdChFromChain-NoSplitCh 2.2}
	\\
	&= \left(2-\delta-\frac{R-1+2\delta}{\alpha}\right) \scdot w_f = R\scdot w_f.
	\label{eqn: fwdChFromChain-NoSplitCh 2.4}
\end{align}
To justify inequality~(\ref{eqn: fwdChFromChain-NoSplitCh 2.2}), we note that
$1-\delta-(R-1+2\delta)/(2\alpha)\ge 0$, by~(\ref{ineq:1-d-frac}), so we can
again apply inequalities $w_j\le w_f$ and  $w_i\le w_f$ from
Lemma~\ref{l:fwdChargeFromChain}(b) and (c).
In the last step~(\ref{eqn: fwdChFromChain-NoSplitCh 2.4})
we used equation~(\ref{eq:forwardCh}).

\smallskip
\noindent
\mycase{2.2} The chain $C$ is singleton.
We assume that $w_i > R\scdot w_e$,
otherwise there is no forward charge from the chain. Then the charge to step $t$ is
\begin{align}
w_j + w_i - R\scdot w_e &\le
				w_j + w_i - R\scdot \frac{w_i + w_j}{2\alpha} 
				\nonumber
				\\
				&=\left(1-\frac{R}{2\alpha}\right)\scdot (w_i+w_j)
				\nonumber
				\\
				&\le \left(1-\frac{R}{2\alpha}\right)\scdot (2w_f) 
				\label{eqn: fwdChFromChain-NoSplitCh 3.3}
				\\
				&= \left(2-\frac{R}{\alpha}\right)\scdot w_f < R\scdot w_f
				\label{eqn: fwdChFromChain-NoSplitCh 3.4}
\end{align}
Inequality~(\ref{eqn: fwdChFromChain-NoSplitCh 3.3}) is valid, because
$w_i\le w_f$ and $w_j\le w_f$, by Lemma~\ref{l:fwdChargeFromChain},
and $1-R/(2\alpha)\ge 0$ by~(\ref{ineq:1-R/2alpha}).
In step~(\ref{eqn: fwdChFromChain-NoSplitCh 3.4})
we used~(\ref{ineq:fwdChFromSingleton}).
\end{proof}

We now analyze how the forward charge from a chain combines with split charges.
First we observe that only the first step from a split-charge pair
may receive a forward charge from a chain.

\begin{lemma}\label{l:distributeChargeAndFwdChTo2ndStep}
If $j$ is charged using a split charge to a pair of steps $t$ and $t'$ (where $t'$ is $t+1$ or $t+2$),
then $t'$ does not receive a forward charge from a chain.
\end{lemma}

\begin{proof}
By Lemma~\ref{obs:splitCharges}(b) we have $k=f$,
which implies that steps $t$ and $t+1$ are not chaining steps.
\end{proof}

\begin{lemma}\label{l:fwdAndDistributeCharge}
If $j$ is charged using a split charge to a pair of steps $t$ and $t'$,
$f'$ is the packet scheduled in $t'$ in \ALG, and 
step $t$ receives a forward charge from a chain $C$,
then the total charge to steps $t$ and $t'$ is
at most $R\scdot (w_f + w_{f'})$.
\end{lemma}

\begin{proof}
First we note that $j$ is expiring in step $t$ by Lemma~\ref{obs:splitCharges}(c)
and $f$ is not expiring in step $t$ by Lemma~\ref{obs:splitCharges}(b), so $f\neq i$.

We claim $w_j < w_e$. Indeed, if $w_j > w_e$, then in step $t-1$ the algorithm would have pending
packets $e$ and $i$, plus packet $j$ (pending or lookahead), that need to be
scheduled in slots $t-1$ and $t$. Since $w_e < w_i$ (because step
$t-1$ is a chaining step) and $w_e < w_j$,  
packet $e$ could not be in the plan in step $t-1$ which is a contradiction.
Therefore $w_j < w_e$.

Let $p_1, p_2, p_3$ be the plan at time $t$. 
We split the proof into two cases,
both having two subcases, one for long chains and one for singleton chains.

\smallskip
\noindent
\mycase{1}
$j$ is charged using a distant split charge or $f'$ gets a full back charge.

We claim that $2\alpha\scdot w_i < w_f + w_{f'}$. Indeed, since $i$ is expiring and pending in step $t$
by Lemma~\ref{obs:chainSteps}(a), we have $w_i\le w_{p_1}$.
As the algorithm scheduled $f=p_2$ by Lemma~\ref{obs:splitCharges}(d)
we get $2\alpha\scdot w_{p_1} < w_f + w_{p_3}$. To prove the claim it remains to show $w_f'\ge w_{p_3}$.

If $j$ is charged using a distant split charge, then by Lemma~\ref{l:distantSplitCh:g<p3}
we have $w_g < w_{p_3}$ and in particular, $g\neq p_3$. Thus $p_3$ is pending and expiring in step $t+2$,
hence $w_{p_3}\le w_{f'}$.
Otherwise, if $j$ is charged using a close split charge, then $f'=g$ gets a full back charge.
Hence $d_g = t+2$. Since also $d_{p_3}=t+2$ and the algorithm chooses the heaviest such packet,
we have $w_{p_3}\le w_{g}$.

The claim follows, since 
\begin{equation}\label{eq:fwdCh+SplitCh1}
2\alpha\scdot w_i \le 2\alpha\scdot w_{p_1} < w_f + w_{p_3} \le w_f + w_{f'}\,.
\end{equation}

\smallskip
\noindent
\mycase{1.1} The chain $C$ is long.
We upper bound the total charge to steps $t$ and $t'$ by
\begin{align}
w_f + w_{f'} + w_j + (1-\delta)\scdot w_i &- (R-1+2\delta)\scdot w_e 
			\nonumber
			\\
			&\le w_f + w_{f'} + (2-R-2\delta)\scdot w_e + (1-\delta)\scdot w_i 
			\nonumber
			\\
			&< w_f + w_{f'} + (2-R-3\delta)\scdot w_i + w_i 
			\label{eqn: fwdChFromChain-split 1.1.3}
			\\
			&= w_f + w_{f'} + w_i 
			\label{eqn: fwdChFromChain-split 1.1.4}
			\\
			&<w_f + w_{f'} + \frac{w_f + w_{f'}}{2\alpha}
			\label{eqn: fwdChFromChain-split 1.1.5}
			\\
			& = R\scdot (w_f + w_{f'})
			\nonumber
\end{align}
We can use $w_e<w_i$ in~(\ref{eqn: fwdChFromChain-split 1.1.3}), because $2-R-2\delta\ge 0$ by~(\ref{ineq:2-R-2delta}).
Equality~(\ref{eqn: fwdChFromChain-split 1.1.4}) follows from $2-R-3\delta = 0$ by~(\ref{eq:2-R-3delta})
and inequality~(\ref{eqn: fwdChFromChain-split 1.1.5}) from~(\ref{eq:fwdCh+SplitCh1}).
In the last step we use~(\ref{eq:splitCharges}).

\smallskip
\noindent
\mycase{1.2} The chain $C$ is singleton.
We suppose that $w_i > R\scdot w_e$,
otherwise there is no forward charge from the chain.
We upper bound the total charge to steps $t$ and $t'$ by
\begin{align}
w_f + w_{f'} + w_j + w_i - R\scdot w_e 
	&\le w_f + w_{f'} + (1-R)\scdot w_e + w_i 
	\nonumber
	\\
	&< w_f + w_{f'} + w_i 
	\nonumber
	\\
	&< w_f + w_{f'} + \frac{w_f + w_{f'}}{2\alpha} 
	\label{eqn: fwdChFromChain-split 1.2.3}
	\\
	&= R \scdot (w_f + w_{f'})\,,
	\nonumber
\end{align}
where we apply Equation~\ref{eq:fwdCh+SplitCh1} in~(\ref{eqn: fwdChFromChain-split 1.2.3}), and we use~(\ref{eq:splitCharges}) in the last step.

\smallskip
\noindent
\mycase{2}
$j$ is charged using a close split charge and $f' = g$ does not get a full back charge.
We have $2\alpha\scdot (w_{p_1}-w_g) < w_f + w_g$ by the definition of the close split charge.
Since $i$ is expiring and pending in step $t$
by Lemma~\ref{obs:chainSteps}(a), we have $w_i\le w_{p_1}$.
Hence $2\alpha\scdot (w_i-w_g) < w_f + w_g$.
This is equivalent to
\begin{equation}\label{eq:fwdCh+SplitCh2}
w_i < \frac{w_f + (2\alpha + 1)\scdot w_g}{2\alpha}\,.
\end{equation}

\smallskip
\noindent
\mycase{2.1} The chain $C$ is long.
We again suppose that $w_i > R\scdot w_e$,
as otherwise there is no forward charge from the chain.
The total charge to steps $t$ and $t'=t+1$ is
\begin{align}
w_f + w_j + (1-\delta)\scdot w_i &- (R-1+2\delta)\scdot w_e 
		\nonumber
		\\
		&\le w_f + (2-R-2\delta)\scdot w_e + (1-\delta)\scdot w_i 
		\nonumber
		\\
		&< w_f + (2-R-3\delta)\scdot w_i + w_i 
		\label{eqn: fwdChFromChain-split 2.1.3}
		\\
		&= w_f + w_i 
		\label{eqn: fwdChFromChain-split 2.1.4}
		\\
		&< w_f + \frac{w_f + (2\alpha + 1)\scdot w_{f'}}{2\alpha}
		\label{eqn: fwdChFromChain-split 2.1.5}
		\\
		&= w_f + w_{f'} + \frac{w_f + w_{f'}}{2\alpha} 
		\nonumber
		\\
		&= R \scdot (w_f + w_{f'})\,,
		\nonumber
\end{align}
We can use $w_e<w_i$ in~(\ref{eqn: fwdChFromChain-split 2.1.3}), because $2-R-2\delta\ge 0$ by~(\ref{ineq:2-R-2delta}).
Then we use $2-R-3\delta = 0$ by~(\ref{eq:2-R-3delta}) in~(\ref{eqn: fwdChFromChain-split 2.1.4}),
Equation~\ref{eq:fwdCh+SplitCh2} in~(\ref{eqn: fwdChFromChain-split 2.1.5}),
and Equation~\ref{eq:splitCharges} in the last step.

\smallskip
\noindent
\mycase{2.2} The chain $C$ is singleton.
We upper bound the total charge to steps $t$ and $t+1$ by
\begin{align}
w_f + w_j + w_i - R\scdot w_e 
	&\le w_f + (1-R)\scdot w_e + w_i 
	\nonumber
	\\
	&< w_f + w_i 
	\nonumber
	\\
	&< w_f + \frac{w_f + (2\alpha + 1)\scdot w_{f'}}{2\alpha}
	\label{eqn: fwdChFromChain-split 2.2.3}
	\\
	&=w_f + w_{f'} + \frac{w_f + w_{f'}}{2\alpha} 
	\nonumber
	\\
	&= R \scdot (w_f + w_{f'})\,,
	\nonumber
\end{align}
where we apply~(\ref{eq:fwdCh+SplitCh2}) in inequality~(\ref{eqn: fwdChFromChain-split 2.2.3}), and~(\ref{eq:splitCharges}) in the last step.
\end{proof}

We now summarize our analysis of \LCalpha.  If $t$ is not in a
split-charge pair, we show upper bounds on the total charge to step
$t$.  For each split-charge pair $(t,t')$, we show upper bounds on the
total charge to both steps $t$ and $t'$.  This is sufficient, since
split-charge pairs are pairwise disjoint by
Lemma~\ref{l:noTwoDistribCharges}, thus summing all the bounds gives
the result.

For each step $t$, we distinguish three cases according to whether 
$t$ is in a split-charge pair and whether $t$ is a chaining step.
In all cases, let $f$ be the packet scheduled at time $t$ in \ALG 
and let $j$ be the packet scheduled at time $t$ in \OPT.

\smallskip
\noindent
\mycase{1} Step $t$ is not chaining
and it is not in a split-charge pair.
Then $t$ receives at most one full charge from a packet $p$ such that $w_p \leq w_f$
(by Lemma~\ref{l:oneFullCharge} and charging rules)
and possibly a forward charge from a chain $C$;
then Lemma~\ref{l:fwdChFromChain-NoSplitCh} shows that the sum of a forward charge from a chain
and a full charge is at most $R\scdot w_f$.

\smallskip
\noindent
\mycase{2} Step $t$ is a chaining step.
Then it does not receive a split charge
or a full charge, by Lemma~\ref{l:noFullorSplitChargeInAChain}.
Lemma~\ref{l:chargingToChain} implies that step $t$ receives a charge of at most 
$R\scdot w_f$. 

\smallskip
\noindent
\mycase{3} $(t,t')$ is a split-charge pair, i.e., $t$ is the first
step of the split-charge pair and $t'=t+1$, or $t'=t+2$.  Thus $j$ is
charged using a split charge. Let $f'$ be the packet scheduled in step
$t'$ in \ALG.
 
By Lemma~\ref{l:distributeChargeAndFwdChTo2ndStep} step $t'$ does not
receive a forward charge from a chain.  If step $t$ also does not
receive a forward charge from a chain, then the total charge to steps
$t$ and $t'$ is at most $R\scdot (w_f + w_{f'})$ by
Lemma~\ref{l:splitChargeUB}.  Otherwise, step $t$ receives a forward
charge from a chain and we apply Lemma~\ref{l:fwdAndDistributeCharge}
to show that the total charge to steps $t$ and $t'$ is again at most
$R\scdot (w_f + w_{f'})$.

%
%
%


\section{A Lower Bound for 2-bounded Instances with Lookahead}\label{sec:lookaheadlb}



In this section, we prove that there is no online algorithm for
{\BDPS} with 1-lookahead that has competitive ratio smaller than
$\onefourth (1+\sqrt{17}) \approx 1.281$, even for $2$-bounded
instances.  The idea of our proof is somewhat similar to the proof of
the lower bound of $\phi$ for
{\BDPS}~\cite{hajek_unit_packets_01,andelman_queueing_policies_03,chin_partial_job_values_03}.


\begin{theorem}
Let $R = \onefourth (1 + \sqrt{17})$.
For each $\varepsilon > 0$, no deterministic online algorithm for {\BDPS} with 1-lookahead
can be $(R - \varepsilon)$-competitive, even for $2$-bounded instances.
\end{theorem}

\begin{proof}
Fix some online algorithm $\calA$ and some $\varepsilon > 0$. We will
show that, for some sufficiently large integer $n$ and sufficiently
small $\delta > 0$, there is a $2$-bounded instance of {\BDPS} with
1-lookahead, parametrized by $n$ and $\delta$, for which the optimal
profit is at least $(R-\varepsilon)$ times the profit of $\calA$.

Our instance will consist of phases $0,\ldots,k$, for some $k\le
n$. In each phase $i<n$ we will release three packets whose weights
will grow roughly exponentially from one phase to next. The number $k$
of phases is determined by the adversary based on the behavior of
$\calA$.

The adversary strategy is as follows. We start with phase $0$. 
Suppose that some phase $i$, where $0\le i < n$,
has been reached. In phase $i$ the adversary releases the following three packets:
\begin{compactitem}
\item A packet $a_i$ with weight $w_{i}$, release time $2i+1$ and deadline $2i+1$, i.e., a tight packet.
\item A packet $b_i$ with weight $w_{i+1}$, release time $2i+1$ and deadline $2i+2$.
\item A packet $c_i$ with weight $w_{i+1}$, release time $2i+2$ and deadline $2i+3$.
\end{compactitem}
(The weights $w_i$ will be specified later.)
Now, if $\calA$ schedules an expiring packet in step $2i+1$ (a tight packet $a_i$ or $c_{i-1}$, which may be pending from the previous phase), 
then the game continues; the adversary will proceed to phase $i+1$. Otherwise, 
the algorithm schedules packet $b_i$, in which case the adversary lets $k = i$ and the game ends. 
Note that in step $2i+2$ the algorithm may schedule only $b_i$ or $c_i$, each having weight $w_{i+1}$.
Also, importantly, in step $2i+1$ the algorithm cannot yet see whether the packets from phase $i+1$ will arrive or not.

If phase $i=n$ is reached, then in phase $n$ the adversary releases 
a single packet $a_n$ with weight $w_{n}$ and release time and deadline $2n+1$, 
i.e., a tight packet.

We calculate the ratio between the weight of packets in an optimal schedule and the weight of packets sent by the algorithm.
Let $S_k = \sum_{i=0}^k w_i$.
There are two cases: either $k < n$, or $k = n$.


\smallskip
\noindent
\mycase{1} $k < n$. In all steps $2i+1$ for $i < k$ algorithm $\calA$
scheduled an expiring packet of weight $w_{i}$
and in step $2k+1$ it scheduled packet $b_k$ of weight $w_{k+1}$. In an even step $2i+2$ for $i\leq k$ it scheduled
a packet of weight $w_{i+1}$. Note that there is no packet scheduled in step $2k+3$.
Overall, $\calA$ scheduled packets of total weight $S_{k-1} + w_{k+1} + S_{k+1} - w_0 = 2S_{k+1} - w_k - w_0$.

The adversary schedules packets of weight $w_{i+1}$ in steps $2i+1$ and $2i+2$ for $i < k$ and
all packets from phase $k$ in steps $2k+1$, $2k+2$ and $2k+3$. In total, 
the optimum has a schedule of weight $2S_{k+1} - 2w_0 + w_k$.
The ratio is 
\begin{equation*}
	R_k = \frac{2S_{k+1}+w_k-2w_0}{2S_{k+1} - w_k - w_0}.
\end{equation*}


\smallskip
\noindent
\mycase{2} $k = n$. As before, in all odd steps $2i+1$ for $i < n$
algorithm $\calA$ scheduled an expiring packet of weight $w_{i}$
and in all even steps $2i+2$ for $i<n$ it scheduled a packet of weight $w_{i+1}$. In the last step $2n+1$ it scheduled a packet of weight $w_n$
as there is no other choice. Overall, the total weight of $\calA$'s schedule is $2S_n - w_0$.

The adversary schedules packets of weight $w_{i+1}$ in steps $2i+1$ and $2i+2$ for $i < n$ and a packet of weight $w_{n}$ in the last step $2n+1$
which adds up to $2S_n - 2w_0 + w_n$. The ratio is
\begin{equation*} 
	\hatR_n = \frac{2S_n+w_n-2w_0}{2S_n - w_0}.
\end{equation*} 

We start with an intuitive explanation which leads to the optimal
setting of weights $w_i$ and the ratio $R$ for the instances of the
type described above. We normalize the instances so that $w_0=1$. We
want to set the weights so that $R_k\geq R-\varepsilon$ for all $k\geq
0$ and $\hatR_n\geq R-\varepsilon$. We first find the weights depending
on $\delta$ such that $R_k=R$ for all $k\geq 1$. Using
$w_k=S_k-S_{k-1}$ for $k\geq 1$ and $w_0=1$, the condition $R_k=R$ for
$k\geq 1$ is rewritten as 
\begin{equation}
\label{eq:Rk}
	R=\frac{2S_{k+1} + S_k-S_{k-1}-2}{2S_{k+1} - S_k +S_{k-1}-1}\,,
\end{equation}
or equivalently as
\begin{equation}
\label{eq:recurrence}
	(2R-2)S_{k+1} -(R+1)S_k + (R+1)S_{k-1} = -(2-R)\,.
\end{equation}
A general solution of this linear recurrence with $S_0=w_0=1$ and a
parameter $\delta$ is
\begin{equation}
\label{eq:solution}
S_k=(\gamma+1)\alpha^k+\delta(\beta^k-\alpha^k)-\gamma\,,
\end{equation}
where $\alpha<\beta$ are the two roots of the characteristic
polynomial of the recurrence $(2R-2)x^2-(R+1)x+(R+1)$ and
$\gamma=(2-R)/(2R-2)$. To justify~(\ref{eq:solution}), a general
solution is $A\alpha^k+B\beta^k-\gamma$ for parameters $A$ and $B$ and
a suitable constant $\gamma$. Considering $A=B=0$, the value
$\gamma=(2-R)/(2R-2)$ follows. Considering the constraint $S_0=1$, we
obtain $A+B=\gamma+1$; our parametrization by $\delta$
in~(\ref{eq:solution}) is equivalent but more convenient for further
analysis.

In our case of $R=\onefourth(1+\sqrt{17})$ a calculation gives
\begin{equation}
\label{eq:parameters}
\alpha = R+\onehalf =\onefourth (3 + \sqrt{17})\,,\quad
\beta = R+1 = \onefourth ( 5 + \sqrt{17})\, \mbox{ and }\quad
\gamma = R = \onefourth(1+\sqrt{17})\,.
\end{equation}
A calculation shows that for $\delta=0$, the solution satisfies
$R_0=R$. We choose a solution with a sufficiently small $\delta>0$
which guarantees $R_0\geq R-\varepsilon$. Since $1<\alpha<\beta$, for
large $n$, the dominating term in $S_n$ is $\delta\beta^n$. Thus
\begin{equation} 
\label{eq:Rn}
	\lim_{n\rightarrow\infty}\hatR_n = 
	\lim_{n\rightarrow\infty}
\frac{2S_n+S_n-S_{n-1}}{2S_n} =
	\lim_{n\rightarrow\infty}
\frac{3\delta\beta^n-\beta^{n-1}}{2\delta\beta^n}=\frac{3\beta-1}{2\beta}=R\,.
\end{equation} 
The last equality is verified by a direct calculation; actually it is
the equation that defines the optimal $R$ for our construction (if
$\beta$ as the root of the characteristic polynomial of the recurrence
is expressed in terms of $R$). 

\medskip

For a formal proof, we set $w_0=1$ and for $i=1,2,\ldots$,
\[
w_i=(\gamma+1)\alpha^{k-1}(\alpha-1)+\delta(\beta^{k-1}(\beta-1)-\alpha^{k-1}(\alpha-1))\,,
\]
where the parameters $\alpha$, $\beta$ and $\gamma$ are given by
(\ref{eq:parameters}) and $\delta>0$ is sufficiently small. By a
routine calculation we verify (\ref{eq:solution}) and
(\ref{eq:recurrence}). Thus $R_k=R$ for $k\geq 1$.
For $R_0$, we first verify that $\delta=0$ would yield $w_1=\alpha$ and
$R_0=R$. By continuity of the dependence of $w_1$ and $R_0$ on
$\delta$, for a sufficiently small $\delta>0$, we have $R_0\geq
R-\varepsilon$; fix such a $\delta>0$. Now, for $n\rightarrow\infty$,
$S_n=\delta\beta^n+O(\alpha^n)=\delta\beta^n(1+o(1))$. Thus, the
calculation (\ref{eq:Rn}) gives $\lim_{n\rightarrow\infty}\hatR_n =
R$.  Consequently, $\hatR_n\geq R-\varepsilon$ for a sufficiently large $n$
of our choice. This defines the required instance and
completes the proof.
\end{proof}


\bibliographystyle{plainurl}
\bibliography{../packets}

\end{document}